\documentclass[11pt,a4paper]{amsart}
\usepackage[foot]{amsaddr}
\usepackage{ifxetex}
\ifxetex
  \usepackage[no-math]{fontspec}
\else
\fi
\usepackage{amsmath}
\usepackage{amsfonts}
\usepackage{amssymb}
\usepackage{amsthm}
\usepackage{fullpage}
\usepackage{microtype}
\usepackage[libertine]{newtxmath}
\usepackage[tt=false]{libertine} 
\usepackage{caption}
\usepackage{bbm}
\usepackage{hyperref, color}
\hypersetup{colorlinks=true,citecolor=blue, linkcolor=blue, urlcolor=blue}
\usepackage[linesnumbered,boxed,ruled,vlined]{algorithm2e}
\usepackage{bm}
\usepackage{bbm}
\usepackage[numbers]{natbib}
\usepackage{xcolor}
\usepackage{enumerate} 
\usepackage{enumitem}
\usepackage{tabularx}
\usepackage{array}
\usepackage{cleveref}
\newcolumntype{L}[1]{>{\raggedright\arraybackslash}p{#1}}
\newcolumntype{C}[1]{>{\centering\arraybackslash}m{#1}}
\newcolumntype{R}[1]{>{\raggedleft\arraybackslash}p{#1}}

\usepackage{makecell}

\usepackage{cleveref}
\usepackage{color}
\usepackage{graphicx}
\usepackage{float}
\usepackage{bbm}

\usepackage{footnote}
\makesavenoteenv{tabular}

\usepackage{thm-restate}
\usepackage{commath}

\newtheorem{theorem}{Theorem}[section]
\newtheorem{claim}[theorem]{Claim}
\newtheorem{corollary}[theorem]{Corollary}
\newtheorem{lemma}[theorem]{Lemma}
\newtheorem{definition}[theorem]{Definition}
\newtheorem{conjecture}[theorem]{Conjecture}

\newtheorem{proposition}[theorem]{Proposition}

\DeclareMathOperator*{\E}{\mathbb{E}}
\newcommand{\Braket}[1]{\langle#1\rangle}
\newcommand{\poly}{\text{poly}}
\newcommand{\W}{\bold{W}}
\newcommand{\I}{\bold{I}}
\newcommand{\N}{\bold{N}}
\newcommand{\cC}{\mathcal{C}}
\newcommand{\cF}{\mathcal{F}}
\newcommand{\floor}[1]{\lfloor#1\rfloor}

\newcommand{\warn}[1]{#1}

\title{\warn{New Distinguishers for Negation-Limited Weak Pseudorandom Functions}}

\date{}

\author{Zhihuai Chen}
\author{Siyao Guo}
\author{Qian Li}
\author{Chengyu Lin}
\author{Xiaoming Sun}

\address[Zhihuai Chen, Qian Li, Xiaoming Sun]{Institute of Computing Technology, Chinese Academy of Sciences.
\textnormal{E-mail: \url{chenzhihuai@ict.ac.cn}, \url{liqian@ict.ac.cn}, and \url{sunxiaoming@ict.ac.cn}}.}
\address[Siyao Guo]{New York University, Shanghai \textnormal{E-mail: \url{siyao.guo@nyu.edu}}.}
\address[Chengyu Lin]{Columbia University \textnormal{E-mail: \url{chengyu@cs.columbia.edu}}.}

\begin{document}

\maketitle

\begin{abstract}
We show how to distinguish circuits with $\log k$ negations (a.k.a $k$-monotone functions) from uniformly random functions in  $\exp\big(\Tilde{O}\big(n^{1/3}k^{2/3}\big)\big)$ time using random samples. The previous best distinguisher, due to the learning algorithm by Blais, Cannone, Oliveira, Servedio, and Tan (RANDOM'15), requires $\exp\big(\Tilde{O}(n^{1/2} k)\big)$ time.

Our distinguishers are based on Fourier analysis on \emph{slices of the Boolean cube}. We show that some ``middle'' slices of negation-limited circuits have strong low-degree Fourier concentration and then we apply a variation of the classic Linial, Mansour, and Nisan ``Low-Degree algorithm'' (JACM'93) on slices. Our techniques also lead to a slightly improved weak learner for negation limited circuits under the uniform distribution.
\end{abstract}

\section{Introduction}

One significant goal in the area of cryptography is to understand how simple cryptography can be.  This motivates the study of low complexity cryptography which explores the possibility of implementing cryptographic primitives in low complexity classes. This line of research inherently lies at the intersection of computational complexity and cryptography. It links core problems in both areas and has become an essential source of new perspectives for both areas.

In this work, we continue this line of research and focus on \emph{pseudorandom functions} (PRFs) in \emph{negation-limited computation}. We start by introducing pseudorandom functions and negation-limited computation before connecting them to explain the main motivation of our work.

\vspace{2ex}
\noindent\textbf{Pseudorandom functions.} Pseudorandom functions (PRFs)~\cite{GGM} are fundamental primitives in symmetric cryptography. In particular, they yield direct solutions to most central goals of symmetric cryptography, such as encryption, authentication and identification. They are well studied in the theoretical community and widely used in practice.

As lightweight (computationally limited) devices become popular, the efficiency of cryptographic implementations
also become increasingly significant.  To obtain a better tradeoff between efficiency and security, a weaker notion of PRFs called \emph{weak pseudorandom functions} (See Definition~\ref{def:wprfs}) has been considered. A distinguisher for a weak PRFs aims to distinguish a random member of the family from a truly random function after observing a number of random samples $(x_1,f_s(x_1)),\dots, (x_m,f_s(x_m))$ where $x_1,\dots,x_m$ are independent uniformly random strings from $\{0,1\}^n$ and $f_s:\{0,1\}^n\rightarrow\{0,1\}$ is the function in question. Weak PRFs suffice for many key applications such as encryption and authentication in symmetric cryptography. More importantly, weak PRFs may allow for significant gains in efficiency.  Akavia et al.~\cite{ABGKM14} pointed out weak PRFs have the potential to bypass the limitations of PRFs in low depth circuits. In particular, they provided candidate weak PRFs in a class of low depth circuits where PRFs provably cannot exist. This raises the following natural questions.
\begin{quote} {\it
		Can weak PRFs bypass the limitations of PRFs in other low complexity classes? }
\end{quote}

Besides cryptography, another important motivation for the study of low complexity PRFs comes from explaining the difficulties of obtaining circuit lower bounds and learning algorithms. We refer interested readers to the survey by Bogdanov and Rosen~\cite{BR17}.

\vspace{2ex}
\noindent\textbf{Negation-Limited Computation.} The power of negations is a mystery in complexity theory. One of the main difficulties in proving lower bounds on circuit size using AND, OR, NOT gates is the presence of negation gates: the best such lower bound is linear, whereas if no negation gates are allowed, exponential lower bounds are known~\cite{razborov1985lower,andreev1985method,alon1987monotone,tardos1988gap,berg1999symmetric,harnik2000higher}. In 1958, Markov~\cite{Mark} observed that every Boolean (even multiple-output) function of $n$ variables can be computed by a circuit with only $\log{n}$ negation gates.  In other words, the potential, possibly exponential, gap between monotone computation and non-monotone computation exists due to as few as $\log{n}$ negations.

Besides circuit complexity, the divide between monotone and non-monotone computation exists in general:  while we usually have a fairly good understanding of the monotone case, many things may fail to hold when negation gates are allowed. Aiming at bridging the gap between monotone and non-monotone computation, a body of recent work studies negation-limited computation from multiple angles including learning~\cite{BCOST15}, cryptography~\cite{GMOR15}, Boolean formulas~\cite{GK15,Ros15}, property testing~\cite{CGGKW17,GKW19}, Boolean function conjectures~\cite{LZ17}. Although the above works extend many results in monotone cases to as many as $O(\log n)$ negations, they also leave open several surprisingly basic questions about a single negation ranging from weak learning algorithms to the structure of their Fourier spectrum. More surprisingly, in the context of property testing, a single negation can be exponentially harder than the monotone case~\cite{CGGKW17,GKW19}.  Our understanding of a single negation remains largely a mystery.

When the circuit size is not of interest, the classes of circuits with $\log k$ negations are captured by the class of so-called $k$-monotone functions where each function in the family can be written as the parity of $k$ monotone function (See Section~\ref{sec:k-mon}). To simplify the presentation, we will use $k$-monotone functions instead of circuits with $\log{k}$ negations in some of our discussions.

\vspace{2ex}
\noindent\textbf{PRFs in Negation-Limited Computation.}  Can pseudorandom functions be computed by a few negations? For pseudorandom functions, we have a fairly good understanding. Guo et al.~\cite{GMOR15} showed that PRFs are inherently highly non-monotone and require $\log{n}-O(1)$ negations, which is optimal up to an additive constant. However, the answer to weak PRFs is unsatisfying. Guo et al.~\cite{GMOR15} observed that weak PRFs cannot be monotone due to the weak learner for monotone functions by Blum et al.~\cite{BABL98}. For general $k$, the best distinguisher, due to Blais et al.~\cite{BCOST15}, runs in time $n^{O(k\sqrt{n})}$.  Therefore even for a single negation (\emph{i.e.}, $k=2$), the best distinguisher runs in time $n^{O(\sqrt{n})}$.

The above results demonstrate two strong separations. In negation-limited computation, weak PRFs have the potential to be much simpler than PRFs: even a single negation may have $n^{O(\sqrt{n})}$ hardness whereas PRFs cannot exist. From the angle of weak PRFs, the hardness gap between even a single negation and monotone can be as large as $n^{O(\sqrt{n})}$.  These separations are our main motivation to connect them together to study negation-limited weak PRFs. 

\subsection{Our Results}
Before presenting our main results, we define \emph{weak pseudorandom functions} and \emph{weak learning under uniform distribution}.
\begin{definition}[Weak Pseudorandom Functions]\label{def:wprfs}
Let $S$ be a distribution over $\{0,1\}^m$ and $\{F_s:\{0,1\}^n\rightarrow\{0,1\}\}$ be a family of functions indexed by string $s$ in the support of $S$. We say $\{F_s\}$ is a $(c,\epsilon)$-secure weak pseudorandom functions (wPRFs) if for every (non-uniform) algorithm $D$ that can be implemented by a circuit of size at most $c$,
\begin{align}\label{ineq:def1}
\abs{\Pr_s[D^{F_s}\mbox{ accepts}]-\Pr_R[D^R\mbox{ accepts}]}\leq \varepsilon,
\end{align}
where $s$ is distributed according to $S$, $R$ is a function sampled uniformly at random from the set of all functions from $\{0,1\}^n$ to $\{0,1\}$, and $D^h$ denotes the execution of $D$ with random oracle access to a Boolean function $h:\{0,1\}^n\rightarrow\{0,1\}$. In other words, the distinguisher $D^h$ only has access to random examples of the form $(x,h(x))$ where $x$ is uniformly distributed over $\{0,1\}^n$. \warn{The two probabilities in (\ref{ineq:def1}) are both also over the random samples $x$'s.}
\end{definition}

\begin{definition}[Weak Learning under Uniform Distribution]
	We say that an algorithm $A$ weakly learns a family of Boolean functions $\cF$ if $A$ can only access uniformly random samples, and $\forall f \in \cF$ it outputs a hypothesis $h$ such that with high probability \warn{(over the random samples and the randomness of $A$)}
	\[ \Pr_{x \sim U_{\{0, 1\}^n}}\left[ f(x) \neq h(x) \right] \le \frac{1}{2}-\frac{1}{\poly(n)}, \]
where $U_{\{0,1\}^n}$ is the uniform distribution over $\{0,1\}^n$.
\end{definition}

A weak learner works slightly better than random guessing. But from this small advantage, if it's non-negligible,
one can naturally derive an efficient distinguisher against random function.
Any weak learner explicitly gives an attack on the weak pseudorandom functions candidates.
Conversely, weak pseudorandom functions are hard to learn.
Our main result is new distinguishers for negation-limited weak pseudorandom functions. Our results hold for inefficient circuits and are stated in terms of $k$-monotone functions.

\begin{restatable}{theorem}{distinguisher}\label{thm:distinguisher}
Any family of $k$-monotone functions can be distinguished from uniformly random functions in $\exp\left(O(n^{1/3}(k\log n)^{2/3})\right)$ time.
In other words, any family of $k$-monotone functions is not a
	  $\left(\exp\left(O(n^{1/3}(k\log n)^{2/3})\right),1/3\right)$-secure
	  weak pseudorandom family.
\end{restatable}
The previous best distinguisher for $k$-monotone weak PRFs is the learning algorithm by Blais et al.~\cite{BCOST15} which runs in $\exp\big(O(n^{1/2}k\log n)\big)$ time. Our result improves an $\Omega\big(n^{1/6}(k\log n)^{1/3}\big)$ factor in the exponent.

Theorem~\ref{thm:distinguisher} implies that exponentially secure weak PRFs requires $\log{n}-O(\log\log{n})$ negations, which is optimal up to an additive $O(\log\log{n})$ term. Therefore, weak PRFs cannot bypass the limitations of PRFs in terms of achieving exponential security.

Theorem~\ref{thm:distinguisher} also implies that $1$ negation functions can be distinguished in $\exp\big(O(n^{1/3}\log^{2/3}{n})\big)$ time. Therefore,  unlike testing $1$ negation (using $1$-sided non-adaptive tester)~\cite{GKW19} and learning $1$ negation to high accuracy~\cite{BCOST15}, distinguishing $1$ negation doesn't suffer from the $\exp(\sqrt{n})$ barrier.

It is natural to ask if we can leverage the distinguisher to a learning algorithm. Our second result gives weak learning algorithms for $k$-monotone functions under the uniform distribution.
\begin{restatable}{theorem}{learner}\label{thm:learner}
	$k$-monotone functions are weakly learnable in time $\exp\left(O\big(k\sqrt{n\log n}\big)\right)$.
\end{restatable}

Our result slightly improves the previous best weak learner due to Blais et al.~\cite{BCOST15}, by a $\Omega\big(\sqrt{\log n}\big)$ factor in the exponent.

\warn{We conjecture that both Theorems \ref{thm:distinguisher} and \ref{thm:learner} are not tight. However, we believe that any further improvement of our results, even for a single negation, require completely new techniques or proving rather hard conjectures which seem out of reach. See Section \ref{sec:discussion} for more details.}

\vspace{2ex}
\noindent\textbf{Our Techniques.}
Blais et al.~\cite{BCOST15} showed a Fourier concentration of $k$-monotone functions on low degree monomials, by bounding the total
influence of $k$-monotone functions. Then they apply the ``Low-Degree Algorithm" established by
Linial, Mansour, and Nisan~\cite{LMN93} to learn $k$-monotone functions. One natural idea to improve their learning algorithm is to show Fourier concentration on lower levels. However,
their influence bound is tight and even for monotone functions, we cannot show  concentration bound on fewer than $\Omega(\sqrt{n})$ levels ~\cite{DFTWW14}, which will require at least $n^{\Omega(\sqrt{n})}$ time by applying the ``Low-Degree Algorithm".

Our main technique is using Fourier analysis on slices~\cite{Filmus,srinivasan2011symmetric}. Although the Fourier concentration on the Boolean cube cannot be improved, we show some ``middle" slices of $k$-monotone functions can have much stronger Fourier concentration.  Then by adapting the ``Low-Degree Algorithm" to the slices, we obtain a distinguisher with significantly improved running time. Our weak learner is a simple variant of the ``Low-Degree Algorithm" on slices.

Fourier analysis on slices has a notion of total influence which allows us to show Fourier concentration on a slice in a similar way. We give an upper bound on the sum of total influences for all ``middle" slices of any $k$-monotone function. It implies the existence of a ``middle" slice function with small total influence, and therefore good concentration. Then we optimize the number of ``middle" slices to be analyzed to get an efficient algorithm.

\vspace{2ex}
\noindent\textbf{Paper Organization.} We begin with the basic notations in Section 2, then present the structural results for $k$-monotone functions in Section 3. In Sections 4 and 5, we present the distinguisher and weaker learner.

\section{Preliminaries}
\warn{In this paper, all the logarithms are base 2.}
\subsection{Alternating Number, Negation Complexity, k-monotone Functions}\label{sec:k-mon}

For any two inputs $x, y \in \{0, 1\}^n$, we say $x \prec y$ iff $x \neq y$ and $x_i \le y_i$ for all $i \in [n]$.
A chain $X = (x^1, x^2, \dots, x^\ell)$ of length $\ell$ is an increasing sequence of inputs in $\{0, 1\}^n$
where $x^i \prec x^{i+1}$ for $i \in [\ell - 1]$.
For a Boolean function $f:\{0,1\}^n\rightarrow \{0, 1\}$, we define the alternating number of $f$
on chain $X$ to be the number of value flips on this chain:
\[a(f, X) = \left| \{i \in [\ell - 1] : f(x^i) \neq f(x^{i+1})\}\right|\]
Let $\cC$ be the set of all chains on $\{0, 1\}^n$, the alternating number of $f$ is
\[a(f) = \max_{X \in \cC} a(f, X)\]
Note that the alternating number of a monotone function is no more than 1.

A celebrated result of Markov connects the alternating number of a Boolean function $f$ to the
negation complexity $\N(f)$ -- the minimum number of negation gates required in any
Boolean $\sf{AND}-\sf{OR}$ circuits to compute $f$.

\begin{theorem}[Markov's Theorem~\cite{Mark}]\label{thm:markov}
Let $f:\{0,1\}^n \rightarrow \{0, 1\}$ be a function which is not identically $0$ with $f(0^n) = 0$,
then $\N(f) = \lceil\log\big(a(f) + 1\big)\rceil - 1$.
\end{theorem}

Blais et al.~\cite{BCOST15} showed decomposition for functions with low alternating number~\cite{BCOST15}.
\begin{theorem}[Blais et al.~\cite{BCOST15}]
Let $f$ be a $k$-alternating function, then $f(x) = h(m_1(x),\dots, m_k(x))$
where $m_i(x)$ is monotone and h is the parity function or its negation. Conversely, any
function of this form is $k$-alternating.
\end{theorem}

The above characterization shows a simple structure for functions with a low alternating number,
which are computable by few negation gates.
To simplify notation, we'll focus on the parity of few monotone functions.

\begin{definition} [$k$-monotone function]
A function $f:\{0,1\}^n\rightarrow\{0,1\}$ is said to be $k$-monotone, if there exist
$k$ monotone functions $g_1, g_2, \dots, g_k$ such that $f = g_1 \oplus g_2 \oplus \cdots \oplus g_k$. 

\end{definition}

\subsection{Orthogonal Basis for Functions over a Slice}
Given a set of strings $A\subseteq{[n]\choose r}:=\{(x_1,\ldots,x_n)\in\{0,1\}^n:\sum_i{x_i}=r\}$, denote its density in this slice by $\mu(A)$, \emph{i.e.}, $\mu(A)=|A|/{n\choose r}$. Define its upper shadow as
\begin{center}
	$\partial^+ A:=\{x\in{[n]\choose r+1}:x\succ y$ for some $y\in A\}$,
\end{center}
and its lower shadow as
\begin{center}
	$\partial^- A:=\{x\in{[n]\choose r-1}:x\prec y$ for some $y\in A\}$.
\end{center}

Filmus~\cite{Filmus} and Srinivasan~\cite{srinivasan2011symmetric} independently introduced an orthogonal basis for functions over a slice of the Boolean hypercube ${[n]\choose r}$, which plays a central role in our proofs. All the following definitions can be found in~\cite{Filmus}. We present them here for the reader's convenience.
\begin{definition}
For $d\leq n/2$, a sequence of length $d$ is a sequence $S=s_1,\ldots,s_d$ of distinct numbers in $[n]$. The set of all sequences of length $d$ is denoted by $\mathcal{S}_{n,d}$, and the set of all sequences is denoted by $\mathcal{S}_{n}$.

For any two disjoint sequences $A,B\in \mathcal{S}_{n,d}$ we define the function $\chi_{A,B}$ as
\[\chi_{A,B}=\prod_{i=1}^d(x_{a_i}-x_{b_i}).\]
\end{definition}
\begin{definition}
  For $d\leq n/2$, let $A,B\in \mathcal{S}_{n,d}$ be disjoint. We say that $A$ is smaller than $B$, written $A<B$, if $a_i<b_i$ for all $i\in[d]$. Similarly, we say that $A$ is at most $B$, written $A\leq B$, if $a_i\leq b_i$ for all $i\in [d]$.

  A sequence $B\in \mathcal{S}_n$ is a top set if $B$ is increasing and for some disjoint sequence $A$ of the same length, $A<B$. The set of top sets of length $d$ is denoted by $\mathcal{B}_{n,d}$, and the set of all top sets is denoted by $\mathcal{B}_{n}$.
\end{definition}
\begin{definition}
For $B\in\mathcal{B}_{n,d}$, define
\[\chi_B=\sum_{A\in \mathcal{S}_{n,d}:A<B}\chi_{A,B}.\]
\end{definition}
\begin{theorem}[Filmus~\cite{Filmus}]\label{thm:def} Let $r\leq n/2$ be an integer, the set $\{\chi_B:B\in \mathcal{B}_{n,d}$ for some $d\leq r \}$ is an orthogonal basis for the vector space of functions over the slice ${[n]\choose r}$. The Young-Fourier expansion of $f:{[n]\choose r}\rightarrow \mathbb{R}$ is the unique representation
\[f=\sum_{B\in \mathcal{B}_{n,d},\ d\leq r}\hat{f}(B)\chi_B,\]
where $\hat{f}(B)=\frac{\Braket{f,\chi_B}}{\|\chi_{B}\|_2^2}$. Here $\Braket{f,g}:=\mathbb{E}_{x\sim U}[f(x)g(x)]$. In addition for $B\in \mathcal{B}_{n,d}$,
\begin{enumerate}
  \item $\| \chi_{B}\|_2^2=\prod_{i=1}^{d}\frac{(b_i-2(i-1))(b_i-2(i-1)-1)}{2}\cdot 2^d\frac{r^{\underline{d}}(n-r)^{\underline{d}}}{n^{\underline{2d}}}=n^{O(d)}$. In particular, if $r\geq \frac{n}{4}$ and $d=o(n)$, then $\|\chi_{B}\|_2^2\geq 2^{-O(d)}$. Here, $r^{\underline{d}}=\prod_{i=0}^{d-1}(r-i)$.
  \item $\| \chi_{B}\|_{\infty}\leq\sum_{A\in \mathcal{S}_{n,d}:A<B}\|\chi_{A,B}\|_{\infty}=n^{O(d)}$.
\end{enumerate}
\end{theorem}

By Boolean duality, we can extend the above Young-Fourier expansion to where $r>n/2$ naturally. This can be done by replacing the basis $\{\chi_B(x)\}$ by $\{\chi_{\bar{B}}(x):=\chi_B(\bar{x})\}$ where $\bar{x}$ is obtained by flipping all bits of $x$.
\begin{corollary}\label{thm:def2}
    Let $r>n/2$ be an integer, the set $\{\chi_{\bar{B}}(x):B\in \mathcal{B}_{n,d}$ for some $d\leq n-r\}$ is an orthogonal basis for the vector space of functions over the slice ${[n]\choose r}$. The Young-Fourier expansion of $f:{[n]\choose r}\rightarrow \mathbb{R}$ is the unique representation
\[f=\sum_{B\in \mathcal{B}_{n,d},\ d\leq n-r}\hat{f}(B)\chi_{\bar{B}},\]
where $\hat{f}(B)=\frac{\Braket{f,\chi_{\bar{B}}}}{\|\chi_{\bar{B}}\|_2^2}$. In addition, we have $\|\chi_{\bar{B}}\|_2^2=\|\chi_{B}\|_2^2$ and $\|\chi_{\bar{B}}\|_\infty=\|\chi_{B}\|_\infty$.
\end{corollary}

Like functions over the Boolean hypercube, we can define the total weight on level $d$:
\begin{definition} Let $f:{[n]\choose r}\rightarrow\mathbb{R}$, for any $r\leq n$, define
\[
\W^d(f) = \sum_{B\in \mathcal{B}_{n,d}}\hat{f}(B)^2\| \chi_B\|_2^2,
\]
and denote $\W^{>d}(f) = \sum_{d' > d} \W^d$ and $\W^{\leq d}(f) = \sum_{d'\leq  d} \W^d$.
\end{definition}
\begin{definition}
Let $f:{[n]\choose r}\rightarrow \{\pm1\}$. For $i,j\in [n]$, define the influence of $f$ on the pair $(i,j)$ as
\[\I_{ij}[f]=2\Pr[f(x^{(i,j)})\neq f(x)].\]
Here $x^{(i,j)}$ is obtained by switching $x_i$ and $x_j$.
The total influence of $f$ is
\[\I[f]=\frac{1}{n}\sum_{1\leq i<j\leq n}\I_{ij}[f].\]
\end{definition}
\begin{lemma}[O'Donnell and Wimmer~\cite{kkl09}]\label{key}
    Let $f:{[n] \choose r}\rightarrow \{\pm1\}$ and $A=\{x\in {[n]\choose r}: f(x)=-1\}$, then
    \[\min\{\mu(\partial^+ A),\mu(\partial^- A)\}\geq \mu(A)+\frac{n}{4r(n-r)}\cdot \I[f]\geq\mu(A)+\frac{1}{n}\I[f].\]
\end{lemma}

\begin{theorem}[Filmus~\cite{Filmus}]\label{slice_inf}
Let $f:{[n]\choose r}\rightarrow \{\pm1\}$. Then
\[\I[f]=\sum_{d \leq \min(r,n-r)}\frac{d(n+1-d)}{n} \cdot \W^d
 = \sum_{B\in \mathcal{B}_{n,d}, d \leq \min(r,n-r)}\frac{d(n+1-d)}{n} \cdot \hat{f}(B)^2\| \chi_B\|_2^2.\]

In addition, according to Parseval's identity,
\[\sum_d \W^d = \sum_{B\in \mathcal{B}_{n,d},d\leq \min(r,n-r)}\hat{f}(B)^2\| \chi_B\|_2^2=\| f\|_2^2=1.\]
\end{theorem}

\subsection{Basic Inequalities}
Finally, we will make use of the Hoeffding bound.
\begin{theorem}[Hoeffding Bound]
	Let $X=\sum_{i=1}^{n}X_i$, where $X_i\in[a_i,b_i]$ are independent random variables. Then for any $\theta>0$,
	\[\Pr(|X-\mathbb{E}(X)|\geq \theta)\leq 2\exp\left(-\frac{2\theta^2}{\Sigma_{i}(b_i-a_i)^2}\right).\]
\end{theorem}

\begin{corollary}\label{Coro:Heff}
	Let $X$ be a random variable with distribution $\mathcal{D}$ whose range is $[l,u]$. Let $X_1,\ldots,X_m$ be its independent samples. Then w.p. $1-\delta$, for any $\epsilon>0$,
    \[\abs{\frac{1}{m}\sum_{i=1}^{m}X_i-\mathbb{E}(X)}\leq \epsilon\]
	as long as $m\geq (u-l)^2\log(2/\delta)/(2\epsilon^2)$.
\end{corollary}

\warn{
The following fact will also be used.
\begin{proposition}\label{prop:binomial}
For $t=o(n)$, ${n\choose n/2-t/2}/2^n=\frac{1}{\sqrt{n}}\cdot 2^{-O(t^2/n)}$.
\end{proposition}
\begin{proof}
By Stirling's approximation,
\[
\log {n\choose n/2-t/2}=\log(1+o(1))+\log\sqrt{\frac{n}{2\pi(n/2-t/2)(n/2+t/2)}}+n\cdot H\left(\frac{n/2-t/2}{n}\right),
\]
where $H(p)=-p\log p-(1-p)\log(1-p)$ is the binary entropy function.
As $t/2n=o(1)$, by the Taylor expansion of the entropy function around $1/2$, we have 
\[
H\left(\frac{1}{2}-\frac{t}{2n}\right)=1-\frac{1+o(1)}{2\ln 2}\left(\frac{t}{n}\right)^2.
\]
The conclusion follows immediately.
\end{proof}
}

\section{Concentration Property of $k$-Monotone Functions}
\warn{In the rest of this paper, for a function $f:\{0,1\}^n\rightarrow\{0,1\}$, we convert the range to $\{\pm 1\}$. 
The mapping from $\{0,1\}$ to $\{-1,1\}$ is given by $1-2b$, sending $0$ to $1$ and $1$ to $-1$. So a function $f:\{0,1\}^n\rightarrow\{\pm 1\}$ is said to be $k$-monotone if $(1-f)/2$ is $k$-monotone.}

In this section, we show some ``middle" slice
of a $k$-monotone function has Fourier concentration. 
For functions $f:\{0,1\}^n\rightarrow\{\pm1\}$, let $f|_r$ be the subfunction of $f$ restricted to ${[n]\choose r}$ and $\mu(f|_r):=\mu(f|_r^{-1}(-1))$.

\begin{definition}[$(t,d,\epsilon)$-concentration]
We say $f:\{0,1\}^n\rightarrow\{\pm1\}$ is $(t,d,\epsilon)$-concentrated if the following holds: for some $r$ such that $n/2-t/2\leq r\leq n/2+t/2$,
$$\W^{>d}(f|_r)=\sum_{B\in \mathcal{B}_{n,d'}: d'> d }\widehat{f|_r}(B)^2\| \chi_B\|_2^2 < \epsilon.$$
\warn{Intuitively, $f$ has low-degree Fourier concentration on at least one of the middle slices.}
\end{definition}

\begin{lemma} \label{lm:main}
Let $f:\{0,1\}^n\rightarrow\{\pm1\}$ be a $k$-monotone function.
For any $1<t \leq n$ and any $d$, $\epsilon$ such that $d\epsilon \geq 2kn/t$,  $f$ is $(t, d, \epsilon)$-concentrated.
\end{lemma}

Lemma \ref{lm:main} follows from an upper bound on the sum of total influences on slices.
\begin{proposition}\label{prop:main}
Let $f:\{0,1\}^n\rightarrow\{\pm1\}$ be a $k$-monotone function. Then $\sum_{r=0}^{n-1} \I[f|_r]\leq kn$.
\end{proposition}

We first prove Lemma~\ref{lm:main} using Proposition~\ref{prop:main}.

\begin{proof}[Proof of Lemma~\ref{lm:main}]
By contradiction, we assume that $\W^{>d}(f|_r)\ge\epsilon\ge \frac{2kn}{d(t-1)}$ for any $n/2-t/2\leq r\leq n/2+t/2$.
  According to Proposition~\ref{prop:main},  we have $\sum_{r=\lceil n/2-t/2\rceil}^{\floor{n/2+t/2}} \I[f|_r]\leq kn$. By averaging, let $n/2 - t/2\leq r\leq n/2 + t/2$ be such that
  $\I[f|_r] \le kn/t$.
  By Theorem~\ref{slice_inf}, we can deduce that,
  \begin{align*}
    \frac{kn}{t} \ge \I[f|_r] & =\sum_{d' \le \min(r,n-r)}\frac{d'(n+1-d')}{n} \cdot \W^{d'}(f|_r)\\
     & \ge \sum_{d < d' \le \min(r,n-r)}\frac{d'(n+1-d')}{n}\cdot \W^{d'}(f|_r) \\
     & \ge \frac{d(n+1-d)}{n} \cdot \sum_{d < d' \le \min(r, n-r)}\W^{d'}(f|_r) \\
     & > \frac{d}{2} \cdot \epsilon\geq \frac{d}{2} \cdot \frac{2kn}{dt}\geq \frac{kn}{t} \; ,
  \end{align*}
a contradiction.
\end{proof}

Now we prove Proposition~\ref{prop:main}.

\begin{proof}[Proof of Proposition~\ref{prop:main}]
    Suppose $f$ is the parity of $h_1,\cdots,h_k$ where each $h_i$ is monotone.
    For any $r$, when we switch $x_i$ and $x_j$, $f|_r(x)$ changes only if at least one $h_i|_r(x)$ changes for $i=1,2,\dots, k$. Thus, combining with the union bound, we have
    \begin{align}\label{eq:sec3_1}
    \I[f|_r]\leq \sum_{i=1}^k\I[h_i|_r].
    \end{align}
    
    \warn{Since $h_i$ is monotone, the upper shadow of $h_i|^{-1}_r(-1)$ is a subset of $h_i|^{-1}_{r+1}(-1)$. Then according to Lemma~\ref{key}, we have
    \[
    \mu(h_i|_{r+1})\geq \mu(h_i|_r)+\frac{1}{n}I(h_i|_r),
    \]
    which implies 
    \begin{align}\label{eq:sec3_2}
    \frac{1}{n}\sum_{r=0}^{n-1} \I[h_i|_r]\leq \mu(h_i|_{n})-\mu(h_i|_{0})\leq 1.
    \end{align}
    
    Inequalities (\ref{eq:sec3_1}) and (\ref{eq:sec3_2}) imply the desired conclusion.}
\end{proof}
 
\section{Distinguishers for $k$-monotone Functions}

In this section, we prove the following theorem.

\distinguisher*

We prove this theorem by giving a distinguisher for $(t,d,1/2)$-concentrated functions.

\begin{proposition}
For $t\leq n/4$ and $d=o(n/\log n)$, any family of $(t,d,1/2)$-concentrated functions can be distinguished from uniform random functions in
$2^{O(d\log n + t^2/n)}$ time.
\end{proposition}

By Lemma~\ref{lm:main}, every $k$-monotone function is \warn{ $\left((kn^2\log{n})^{1/3},4\big(\frac{k^2n}{\log n}\big)^{1/3},1/2\right)$}-concentrated, then Theorem \ref{thm:distinguisher} follows. Now we prove the proposition.

\begin{proof}
The distinguisher is given in Algorithm~\ref{alg:distinguisher}. We'll show that
\begin{itemize}
    \item ({\bf Soundness}) It accepts a uniform random function w.p. $o(1)$;
    \item ({\bf Completeness}) It accepts any $(t,d,1/2)$-concentrated function w.p. $1-o(1)$;
    \item ({\bf Complexity}) Its sample/time complexity is $2^{O(d\log n + t^2/n)}$.
\end{itemize}

\begin{algorithm}[t]
\caption{A Distinguisher for $(t,d,1/2)$-Concentrated Functions}\label{alg:distinguisher}

    Let $C$ be a large enough constant\;
    \For{$r\leftarrow \lceil\frac{n}{2}-\frac{t}{2}\rceil$ \KwTo $\lfloor\frac{n}{2}+\frac{t}{2}\rfloor$}
 {
     $S\leftarrow 0$\;
     \For{$B\in \mathcal{B}_{n,d'}$ with $d'\leq d$}
     {
          \If {$r\leq n/2$}
           {Estimate $\Braket{f|_r, \chi_B}$ with accuracy $n^{-C\cdot d}$\;
           $S\leftarrow S+\widehat{f|_r}(B)^2\| \chi_B\|_2^2$\;
          \tcp{$\widehat{f|_r}(B)^2\| \chi_B\|_2^2 = \Braket{f|_r, \chi_B}^2/\|\chi_B\|_2^2$}}
          \Else {Estimate $\Braket{f|_r, \chi_{\bar{B}}}$ with accuracy $n^{-C\cdot d}$\;
           $S\leftarrow S+\widehat{f|_r}(\bar{B})^2\| \chi_{\bar{B}}\|_2^2$\;
          \tcp{$\widehat{f|_r}(B)^2\| \chi_{\bar{B}}\|_2^2= \Braket{f|_r, \chi_{\bar{B}}}^2/\|\chi_{\bar{B}}\|_2^2$}}
     }
     \If{$S\geq 3/8$}
        {Return True\;}
 }
 Return False\;

\end{algorithm}

\vspace{2ex}
\noindent\textbf{Soundness.}\quad Let $f$ be a uniform random function. We claim that for each $\frac{n}{2}-\frac{t}{2}\leq r\leq \frac{n}{2}+\frac{t}{2}$, the variable $S$ in Line 11 is at most $1/4$ w.p. $1-o(\frac{1}{n})$, which concludes the soundness by the union bound. 

Fix such an $r$. W.l.o.g., we assume that $r\leq n/2$. For any $B\in\mathcal{B}_{n,d'}$ where $d'\leq d$, it is easily seen that $\E_f\big[\Braket{f|_r, \chi_B}\big]=0$, then by the Hoeffding bound,
\[
    \Pr_f\left[\left|\Braket{f|_r, \chi_B}\right|\ge \theta\right]\leq 2 \exp\left(-2\theta^2 {n \choose r} \big/ \|\chi_B\|^2_2\right)\leq 2\exp\left(-\theta^2 \left(\frac{4}{3}\right)^{n/4-o(n)}\right),
\]
where the last inequality is due to that ${n \choose r} \geq (\frac{n}{r})^r \geq (\frac{4}{3})^{n/4}$ and $\|\chi_B\|^2_2=2^{O(d\log n)}=(\frac{4}{3})^{o(n)}$.
In particular, by letting $\theta=\left(\frac{3}{4}\right)^{n/10}$ and using the union bound, we have that with probability at least $1-n^{O(d)}\cdot \exp\left(-\theta^2 \left(\frac{4}{3}\right)^{n/4-o(n)}\right)=1-o(\frac{1}{n})$, $\big|\Braket{f|_r, \chi_B}\big|< \theta$ for every $B\in\mathcal{B}_{n,d'}$ where $d'\leq d$. Thus, w.p. $1-o(\frac{1}{n})$,
\begin{align*}
    \W^{\leq d}[f|_r]&=\sum_{B\in\mathcal{B}_{n,d'},d'\leq d} \widehat{f|_r}(B)^2\|\chi_B\|_2^2=\sum_{B\in\mathcal{B}_{n,d'},d'\leq d} \frac{\Braket{f|_r, \chi_B}^2}{\|\chi_B\|_2^2}
 \leq \sum_{B\in\mathcal{B}_{n,d'},d'\leq d} \frac{\theta^2}{\|\chi_B\|_2^2}\leq \frac{1}{8},
\end{align*}
where the last inequality holds for sufficiently large $n$. Finally, $S$ is an estimate of $\W^{\leq d}[f|_r]$ with additive error $n^{-\Omega(d)}$.

\vspace{2ex}
\noindent\textbf{Completeness.}\quad Let $f$ be a $(t,d,1/2)$-concentrated function. By definition, there is some $r$ such that $n/2-t/2\leq r\leq n/2+t/2$ and $W^{\leq d}[f|_r]>1/2$. As $S$ is an estimate of $\W^{\leq d}[f|_r]$ with additive error $n^{-\Omega(d)}$, we conclude that Algorithm~\ref{alg:distinguisher} accepts $f$ with high probability.

\vspace{2ex}
\noindent\textbf{Complexity.}\quad The loop in Line~2 is repeated at most $t$ times. In Line~4, the number of strings $B\in \mathcal{B}_n$ of length at most $d$ we enumerated is at most $n^{O(d)}$.
Furthermore, for each $n/2-t/2\leq r\leq n/2+t/2$ and each $B\in\mathcal{B}$ of size $\leq d$, according to the Hoeffding bound, $n^{O(d)}$ uniform random samples on the slice ${[n]\choose r}$ are sufficient to estimate $\Braket{f|_r, \chi_B}$ with accuracy $n^{-C\cdot d}$. In addition, a random uniform sample is from the slice ${[n]\choose r}$ with probability ${n\choose r}/2^n$, which is $\frac{1}{\sqrt{n}}\cdot2^{-O(t^2/n)}$ according to Proposition \ref{prop:binomial}. Thus, the total number of random samples used is at most $t \cdot n^{O(d)}\cdot n^{O(d)}\cdot 2^{O(t^2/n)}=2^{O(d\log n+t^2/n)}$.

Besides, the function $\chi_B=\sum_{A\in \mathcal{S}_{n,d}:A<B}\chi_{A,B}$ can be computed by enumerating all $n^{O(d)}$ strings $A$ in $\mathcal{S}_{n,d}$.
Thus, the time complexity is also $2^{O(d\log n + t^2/n)}$.
\end{proof}

\section{Weak Learners for $k$-monotone Functions}
In this section, we prove the following theorem.

\learner*

We prove Theorem \ref{thm:learner} by giving a weak learner for $(t,d,1/2)$-concentrated functions. By Lemma \ref{lm:main}, $k$-monotone functions are $\left(\sqrt{n\log n},4k\sqrt{\frac{n}{\log n}},1/2\right)$-concentrated, then Theorem \ref{thm:learner} follows.

\begin{proposition}\label{prop:learner}
    For $t=O(\sqrt{n\log n})$ and $d=o(n/\log n)$, Algorithm \ref{alg:learner} weakly learns $(t,d,1/2)$-concentrated functions in $2^{O(d\log n + t^2/n)}$ time.
\end{proposition}

To learn $(t,d,1/2)$-concentrated functions $f$, Algorithm~\ref{alg:learner} tries to find out the slice ${[n] \choose r}$ on which $f|_r$ is concentrated, and then figures out a function $g:{[n]\choose r}\rightarrow \mathbb{R}$ which is very close to $f|_r$. To convert the approximated function $g$ to a Boolean-valued function, we can utilize Claim~\ref{lem:learn} similar to Exercise 3.34 in \cite{book}.
For the rest of the slices, the learner just outputs the most frequent value. \warn{Since $t=O(\sqrt{n\log n})$, each slice in $[n/2-t/2, n/2+t/2]$ is at least a $\frac{1}{\sqrt{n}}\cdot 2^{-O(t^2/n)}=1/\poly(n)$ fraction according to Proposition \ref{prop:binomial}. Hence we get a $(1/2-1/\poly(n))$-close function $h$.}

\begin{algorithm}[t]
\caption{A weak learner for $(t,d, 1/2)$-concentrated functions}
\label{alg:learner}
  Let $C$ be a large enough constant\;
 \For{$r\leftarrow \frac{n}{2}-\frac{t}{2}$ \KwTo $\frac{n}{2}+\frac{t}{2}$}
 {
     $S\leftarrow 0$\;
     $p\leftarrow 1$\;
     \For{$B\in \mathcal{B}_{n,d'}$ with $d'\leq d$}
     {
        \If{$r\leq n/2$}
        {
            Estimate $\Braket{f|_r, \chi_B}$ with accuracy $n^{-C\cdot d}$\;
            $S\leftarrow S+\widehat{f|_r}(B)^2\| \chi_B\|_2^2$\;
        }\Else
        {
            Estimate $\Braket{f|_r, \chi_{\bar{B}}}$ with accuracy $n^{-C\cdot d}$\;
            $S\leftarrow S+\widehat{f|_r}(B)^2\| \chi_{\bar{B}}\|_2^2$\;
        }
          
     }
     \If{$S\geq 3/8$}
       {
       \If{$r\leq n/2$}{
            $g(x)\leftarrow \sum_{B\in\mathcal{B}_{n,d'}:d'\leq d}\widehat{f|_r}(B)\chi_B(x)$\; 
       }\Else{
            $g(x)\leftarrow \sum_{B\in\mathcal{B}_{n,d'}:d'\leq d}\widehat{f|_r}(B)\chi_{\bar{B}}(x)$\;
       }

        \While{$p>\sqrt{3}/4$}
        {
            Pick $\theta \in [-1,1]$ uniformly at random\;
            Estimate $p\leftarrow\Pr[f|_r\neq\mbox{sgn}(g-\theta)]$\;
        }
        Estimate $\mu_{\neq r}\leftarrow \mathbb{E}_{x}[f(x)\mid |x|\neq r]$\;

        Return $h(x)=\begin{cases}
                         \mbox{sgn}(g(x)-\theta) & \mbox{if } |x|=r;\\
                         \mbox{sgn}(\mu_{\neq r}) & \mbox{if } |x|\neq r.
                     \end{cases}$
        }
  
 }
 Return $h(x)\equiv 0$.
\end{algorithm}

\begin{proof}[Proof of Proposition~\ref{prop:learner}]
We first show that Algorithm \ref{alg:learner} weakly learns $(t,d,1/2)$-concentrated functions. Let $f$ be a $(t,d,1/2)$-concentrated function. For each $n/2-t/2\leq r\leq n/2+t/2$, the variable $S$ in Line 12 is an estimate of $\W^{\leq d}[f|_r]$ with additive error $n^{-\Omega(d)}$. Then for some $r^\star$, the condition $S\geq 3/8$ in Line 12 holds, and Algorithm \ref{alg:learner} executes Lines 13-21. 
For the function $g$ obtained in Line 14 or Line 16, and sufficiently large $n$, 
\begin{align*}
  ||f|_{r^\star}-g||_1\leq ||f|_{r^\star}-g||_2 = \sqrt{\sum_{B\in \mathcal{B}_{n,d'}, d'\leq d}{(\widehat{f|_r}(B)-\hat{g}(B))^2}\|\chi_B\|_2^2+W^{>d}[f|_{r^{\star}}]}\leq \sqrt{o(1)+5/8}\leq \sqrt{3}/2.
\end{align*}
To convert $g$ to a Boolean-valued function, we utilize the following claim. 
\begin{claim}\label{lem:learn}
 Suppose $f:{[n]\choose r}\rightarrow\{-1,1\}$ and $g:{[n]\choose r}\rightarrow \mathbb{R}$. Pick $\theta \in [-1,1]$ uniformly at random and define $g'=\mbox{sgn}\big(g(x)-\theta\big)$, we have $\mathbb{E}_\theta\left[\Pr_x\big(f(x)\neq g'(x)\big)\right]\leq \|f-g\|_1/2$.
\end{claim}
\begin{proof} By rewriting the last formula and swapping the expectation operators, we have
\begin{equation*}
     \begin{aligned}
        \E_\theta \left[\Pr_x\big(f(x)\neq g'(x)\big)\right] &= \E_\theta \E_x \left[1_{f(x) \neq g'(x)}\right] = \E_x\E_\theta \left[1_{f(x) \neq g'(x)}\right]
        = \E_x\left[\Pr_\theta \big(f(x) \neq \mbox{sgn}(g(x)-\theta)\big)\right]\\
          &\leq\E_x\left[\frac{|f(x)-g(x)|}{2}\right]=\frac{\|f-g\|_1}{2}.\qedhere
    \end{aligned}
\end{equation*}
\end{proof}
Thus, for a random $\theta\in[-1,1]$, $\Pr\big[f|_r(x)\neq \mbox{sgn}(g(x)-\theta)\big] \leq \sqrt{3}/4$ holds with a constant probability. That is, with high probability, the loop of Lines 17-19 is repeated a constant number of times, and we will get a $\theta^\star$ such that $\Pr\big[f|_r(x)\neq \mbox{sgn}(g(x)-\theta^\star)\big] \leq \sqrt{3}/4$. Finally, we have
\begin{align*}
\Pr[h(x)\neq f(x)]=&\Pr[|x|\neq r]\Pr[f(x)\neq \mbox{sgn}(\mu_{\neq r})\mid |x|\neq r]+\Pr[|x|=r]\Pr\big[f|_r(x)\neq \mbox{sgn}(g(x)-\theta^\star)\big]\\
\leq &\left(1-\frac{{n \choose r}}{2^n}\right)\cdot \frac{1}{2}+ \frac{{n \choose r}}{2^n} \cdot \frac{\sqrt{3}}{4} =\frac{1}{2} + \frac{{n \choose r}}{2^n}\cdot\left(\frac{\sqrt{3}}{4}-\frac{1}{2}\right)\\=&\frac{1}{2}-\left(\frac{1}{2}-\frac{\sqrt{3}}{4}\right)\cdot\frac{1}{\sqrt{n}}\cdot 2^{-O(t^2/n)}=\frac{1}{2} - \frac{1}{\mbox{poly}(n)},
\end{align*}
\warn{where the second last equality is according to Proposition \ref{prop:binomial} and the last equality is due to that $t=O(\sqrt{n\log n})$.}

What remains is to show that Algorithm \ref{alg:learner} terminates in $2^{O(d\log n + t^2/n)}$ time. First, as shown in the analysis of Algorithm \ref{alg:distinguisher}, for each $n/2-t/2\leq r\leq n/2+t/2$, it costs $2^{O(d\log n + t^2/n)}$ time to execute Lines 5-11. For $r^\star$, Algorithm \ref{alg:learner} would execute Lines 13-21. As shown above, the loop of Lines 17-19 is  repeated a constant of times. So, it costs $2^{O(d\log n + t^2/n)}$ time to execute Lines 13-21. Therefore the total time complexity is $t\cdot 2^{O(d\log n + t^2/n)}=2^{O(d\log n + t^2/n)}$.
\end{proof}

\section{Discussion and Open Problems}\label{sec:discussion}

\noindent\textbf{Fourier analysis on slices.}\quad It is surprising to us that a simple variant of the ``Low-Degree Algorithm'' on slices can outperform the classic ``Low-Degree Algorithm'' in terms of attacking negation-limited weak PRFs. To the best of our knowledge, unlike Fourier analysis on the Boolean cube, Fourier analysis on slices has not been explored in cryptography. It is an extremely interesting direction to use this technique to attack more cryptographic constructions, particularly ones which are secure against attacks based on standard Fourier analysis. 

\vspace{2ex}
\noindent\textbf{The hardness of $1$ negation weak PRFs.} One of the most intriguing open problems is how hard can $1$ negation weak PRFs be? Our bound suggests that, unlike testing $1$ negation (using a 1-sided non-adaptive tester)~\cite{GKW19} and learning $1$ negation to high accuracy~\cite{BCOST15}, distinguishing $1$ negation is significantly more efficient than $2^{O(\sqrt{n})}$. Can we have polynomial time distinguishers? We believe that new structural results of $2$-monotone functions are required for polynomial time distinguishers.  

\vspace{2ex}
\noindent\textbf{Fourier spectrum of $k$-monotone functions on low levels.} It is a basic fact \cite{book} that every monotone function $f:\{0,1\}^n\rightarrow\{0,1\}$ has a large Fourier coefficient on the first two levels. Does a similar statement hold for $k$-monotone functions? In particular, we are curious about the following conjecture.
\begin{conjecture}\label{conj:fourier_spectrum}
Let $f:\{0,1\}^n\rightarrow \{0,1\}$ be a $k$-monotone function. There exists a set $S\subseteq[n]$ of size at most $k$ such that $|\hat{f}(S)|=1/\mathrm{poly}(n)$.
\end{conjecture}
\warn{Conjecture \ref{conj:fourier_spectrum} immediately implies an efficient weak learner (and an efficient distinguisher) for $k$-monotone functions. In fact, our first attempt to distinguish $k$-monotone is to prove Conjecture \ref{conj:fourier_spectrum}. So far, even the following much weaker conjecture remains open. 
\begin{conjecture}\label{conj:fourier_spectrum2}
	Let $f:\{0,1\}^n\rightarrow \{0,1\}$ be a $2$-monotone function. There exists a set $S\subseteq[n]$ of size $o(\sqrt{n})$ such that $|\hat{f}(S)|>0$.
\end{conjecture}
}

\vspace{4ex}
\noindent\textbf{Acknowledgments.} We sincerely thank the anonymous reviewers for their detailed and constructive comments. We thank Shengyu Zhang for fruitful discussions at the early stage of this work. Siyao Guo would like to thank Igor Carboni Oliveira for telling her Conjectures~\ref{conj:fourier_spectrum} and Conjecture~\ref{conj:fourier_spectrum2}.

\newpage

\bibliographystyle{alpha}
\bibliography{paper}

\newcommand{\etalchar}[1]{$^{#1}$}
\begin{thebibliography}{DSFT{\etalchar{+}}15}

\bibitem[AB87]{alon1987monotone}
Noga Alon and Ravi~B Boppana.
\newblock The monotone circuit complexity of boolean functions.
\newblock {\em Combinatorica}, 7(1):1--22, 1987.

\bibitem[ABG{\etalchar{+}}14]{ABGKM14}
Adi Akavia, Andrej Bogdanov, Siyao Guo, Akshay Kamath, and Alon Rosen.
\newblock Candidate weak pseudorandom functions in {$AC^0\circ MOD_2$}.
\newblock In {\em Proceedings of the 5th conference on Innovations in
  theoretical computer science}, pages 251--260. ACM, 2014.

\bibitem[And85]{andreev1985method}
Alexander~E Andreev.
\newblock On a method for obtaining lower bounds for the complexity of
  individual monotone functions.
\newblock {\em Doklady Akademii Nauk SSSR}, 282:1033--1037, 1985.

\bibitem[BBL98]{BABL98}
Avrim Blum, Carl Burch, and John Langford.
\newblock On learning monotone boolean functions.
\newblock In {\em Foundations of Computer Science, 1998. Proceedings. 39th
  Annual Symposium on}, pages 408--415. IEEE, 1998.

\bibitem[BCO{\etalchar{+}}15]{BCOST15}
Eric Blais, Cl{\'{e}}ment~L. Canonne, Igor~Carboni Oliveira, Rocco~A. Servedio,
  and Li{-}Yang Tan.
\newblock Learning circuits with few negations.
\newblock In {\em Approximation, Randomization, and Combinatorial Optimization.
  Algorithms and Techniques, {APPROX/RANDOM} 2015, August 24-26, 2015,
  Princeton, NJ, {USA}}, pages 512--527, 2015.

\bibitem[BR17]{BR17}
Andrej Bogdanov and Alon Rosen.
\newblock Pseudorandom functions: Three decades later.
\newblock In {\em Tutorials on the Foundations of Cryptography.}, pages
  79--158. 2017.

\bibitem[BU99]{berg1999symmetric}
Christer Berg and Staffan Ulfberg.
\newblock Symmetric approximation arguments for monotone lower bounds without
  sunflowers.
\newblock {\em Computational Complexity}, 8(1):1--20, 1999.

\bibitem[CGG{\etalchar{+}}17]{CGGKW17}
Cl{\'{e}}ment~L. Canonne, Elena Grigorescu, Siyao Guo, Akash Kumar, and Karl
  Wimmer.
\newblock Testing $k$-monotonicity.
\newblock In {\em 8th Innovations in Theoretical Computer Science Conference,
  {ITCS} 2017, January 9-11, 2017, Berkeley, CA, {USA}}, pages 29:1--29:21,
  2017.

\bibitem[DSFT{\etalchar{+}}15]{DFTWW14}
Dana Dachman-Soled, Vitaly Feldman, Li-Yang Tan, Andrew Wan, and Karl Wimmer.
\newblock Approximate resilience, monotonicity, and the complexity of agnostic
  learning.
\newblock In {\em Proceedings of the twenty-sixth annual ACM-SIAM symposium on
  Discrete algorithms}, pages 498--511. Society for Industrial and Applied
  Mathematics, 2015.

\bibitem[Fil16]{Filmus}
Yuval Filmus.
\newblock An orthogonal basis for functions over a slice of the boolean
  hypercube.
\newblock {\em Electr. J. Comb.}, 23(1):P1.23, 2016.

\bibitem[GGM86]{GGM}
Oded Goldreich, Shafi Goldwasser, and Silvio Micali.
\newblock How to construct random functions.
\newblock {\em Journal of the ACM}, 33(4):792--807, 1986.

\bibitem[GK17]{GK15}
Siyao Guo and Ilan Komargodski.
\newblock Negation-limited formulas.
\newblock {\em Theoretical Computer Science}, 660:75--85, 2017.

\bibitem[GKW19]{GKW19}
Elena Grigorescu, Akash Kumar, and Karl Wimmer.
\newblock Flipping out with many flips: Hardness of testing {$k$}-monotonicity.
\newblock {\em {SIAM} J. Discret. Math.}, 33(4):2111--2125, 2019.

\bibitem[GMOR15]{GMOR15}
Siyao Guo, Tal Malkin, Igor~C Oliveira, and Alon Rosen.
\newblock The power of negations in cryptography.
\newblock In {\em Theory of Cryptography Conference}, pages 36--65. Springer,
  2015.

\bibitem[HR00]{harnik2000higher}
Danny Harnik and Ran Raz.
\newblock Higher lower bounds on monotone size.
\newblock In {\em Proceedings of the thirty-second annual ACM symposium on
  Theory of computing}, pages 378--387. ACM, 2000.

\bibitem[LMN93]{LMN93}
Nathan Linial, Yishay Mansour, and Noam Nisan.
\newblock Constant depth circuits, {F}ourier transform, and learnability.
\newblock {\em Journal of the ACM (JACM)}, 40(3):607--620, 1993.

\bibitem[LZ17]{LZ17}
Chengyu Lin and Shengyu Zhang.
\newblock Sensitivity conjecture and log-rank conjecture for functions with
  small alternating numbers.
\newblock In {\em 44th International Colloquium on Automata, Languages, and
  Programming, {ICALP} 2017, July 10-14, 2017, Warsaw, Poland}, pages
  51:1--51:13, 2017.

\bibitem[Mar58]{Mark}
Andrey~A Markov.
\newblock On the inversion complexity of a system of functions.
\newblock {\em Journal of the ACM (JACM)}, 5(4):331--334, 1958.

\bibitem[O'D14]{book}
Ryan O'Donnell.
\newblock {\em Analysis of Boolean Functions}.
\newblock Cambridge University Press, 2014.

\bibitem[OW09]{kkl09}
Ryan O'Donnell and Karl Wimmer.
\newblock {KKL, Kruskal-Katona, and monotone nets}.
\newblock In {\em 50th Annual {IEEE} Symposium on Foundations of Computer
  Science, {FOCS} 2009, October 25-27, 2009, Atlanta, Georgia, {USA}}, pages
  725--734, 2009.

\bibitem[Raz85]{razborov1985lower}
Alexander~A Razborov.
\newblock Lower bounds for the monotone complexity of some boolean functions.
\newblock In {\em Soviet Math. Dokl.}, volume~31, pages 354--357, 1985.

\bibitem[Ros15]{Ros15}
Benjamin Rossman.
\newblock Correlation bounds against monotone {$NC^1$}.
\newblock In {\em LIPIcs-Leibniz International Proceedings in Informatics},
  volume~33. Schloss Dagstuhl-Leibniz-Zentrum fuer Informatik, 2015.

\bibitem[Sri11]{srinivasan2011symmetric}
Murali~K Srinivasan.
\newblock {Symmetric chains, Gelfand--Tsetlin chains, and the Terwilliger
  algebra of the binary Hamming scheme}.
\newblock {\em Journal of Algebraic Combinatorics}, 34(2):301--322, 2011.

\bibitem[Tar88]{tardos1988gap}
{\'E}va Tardos.
\newblock The gap between monotone and non-monotone circuit complexity is
  exponential.
\newblock {\em Combinatorica}, 8(1):141--142, 1988.

\end{thebibliography}

\end{document}